\documentclass[letterpaper,10pt,conference]{ieeeconf}

\IEEEoverridecommandlockouts
\usepackage{amsmath}
\usepackage{amsfonts}
\usepackage{amssymb}

\usepackage{amsthm}
\usepackage{graphicx}
\usepackage{algorithm}
\usepackage{algorithmic}
\usepackage[english]{babel}
\usepackage{hyperref}
\usepackage{subcaption}

\newtheorem{theorem}{Theorem}
\newtheorem{lemma}{Lemma}
\newtheorem{proposition}{Proposition}

\theoremstyle{definition}

\theoremstyle{remark}

\title{\LARGE \bf
Minimum-Energy Control for Stochastic Linear Systems
}

\author{
Vaishnavi Sharma$^{1}$ and Vaibhav Katewa$^{2}$
\thanks{The authors are with the
Robert Bosch Centre for Cyber-Physical Systems,
Indian Institute of Science, Bengaluru, India.
{\tt\small (email: $^{1}$svaishnavi@iisc.ac.in,
$^{2}$vkatewa@iisc.ac.in}).}
}

\begin{document}

\maketitle
\thispagestyle{empty}
\pagestyle{empty}
\begin{abstract}
The open-loop minimum-energy control (MEC) problem asks for control inputs that minimize the control energy while ensuring the terminal state reaches a desired value. In this paper, we study the open-loop MEC problem for stochastic linear time-invariant systems. We consider a variety of stochastic constraints that enforce the terminal error to be \emph{small}, and show that Stochastic MEC (SMEC) problem can be formulated as a QCQP. We obtain the optimal solution for SMEC and show that it generalizes the well-known Deterministic MEC problem by introducing novel stochastic variants of controllability and Gramian matrices. Finally, we show that there exists a fundamental lower bound on the optimal control energy which cannot be overcome by increasing the time horizon of the problem. 
\end{abstract}
\section{Introduction}
Minimum-energy control is a classic \textit{open-loop} control problem that involves steering a dynamical system from an initial state to a final desired state within some finite horizon $T$ with minimum control energy (often represented using weighted norm of input $u$). 
For deterministic linear systems, this problem admits a closed-form solution, and has been extensively studied in the literature \cite{brockett1970finite}.\\
Prior works have studied stochastic variants of this problem. For example, in \cite{chen2016optimal}, the problem of steering the state of a linear stochastic system from one Gaussian distribution to another is solved, and closed-loop optimal control is obtained. \cite{cov-steer} solves the problem of steering a stochastic linear system between two desired Gaussian distributions while restricting to the feedback controller to the affine  class.  Further, the authors in \cite{wasserstein} enforce a constraint on the terminal state distribution using Gromov-Wasserstein distance for affine feedback policies. In another work \cite{chen}, a variant of the LQG problem is studied in which penalties on the terminal state are replaced by the specification of the terminal state distribution.
A recent work \cite{reachability} presents a unified framework to find probabilistic reachable sets for stochastic systems by decoupling stochastic and deterministic disturbances. Also, data-driven approaches to find minimum-energy control when accurate dynamics are unavailable have been developed in the deterministic setting \cite{baggio2019data}.\\ All the above works (except \cite{baggio2019data}) focus  on obtaining the optimal \textit{closed-loop controller} with different objectives. In contrast, we focus on obtaining the optimal \textit{open-loop control inputs} in the presence of terminal constraints. To the best of our knowledge, ours is the first paper to address this problem.

Note that for the linear Gaussian systems with affine feedback, closed-loop control inputs are random and can alter both mean and covariance of the state. However, the open-loop control inputs are treated as deterministic, and thus, can only alter the mean. Further, enforcing terminal constraint in stochastic systems is non-trivial since the terminal state is random and the notion of exact reachability is no longer meaningful.

The main contributions of the paper are as follows. First, we show that the Stochastic MEC (SMEC) problem can be formulated as a QCQP for a wide variety of terminal constraints. We then obtain an explicit feasibility condition and explicit solution to the QCQP along with its spectral representation. Second, we introduce the notion of stochastic Gramian and stochastic Controllability matrix for open-loop SMEC problem, which is a counterpart of the standard Gramian for deterministic systems. Using this, we show that the solution structure of Stochastic and Deterministic MEC problems are similar, and in fact, the DMEC problem is a special case of SMEC problem. Third, we show that there exists a fundamental limitation on the minimum achievable cost for the SMEC problem that cannot be overcome by increasing the time horizon.

\noindent \emph{Notation}: 
$A^\top$ denotes the transpose of $A$, and $A^T$ denotes the $T^{\text{th}}$ power of $A$.
We write $A \succ 0$
($A \succeq 0$) when $A$ is positive definite (positive semidefinite).
The largest and smallest eigenvalues of a symmetric matrix $A$ are denoted by
$\lambda_{\max}(A)$ and $\lambda_{\min}(A)$, respectively. $A^\dagger$ denotes the
Moore--Penrose pseudoinverse, $A\otimes B$ the Kronecker product. 
$\lambda_i(A)$ denotes an eigenvalue of $A$. 
Finally, $\mathcal{N}(\cdot,\cdot)$ denotes the Normal distribution.
\section{Problem Formulation}
\label{sec: problem}
We consider a discrete-time, stochastic, linear time-invariant dynamical system given by
\begin{align}
   x(t+1) = Ax(t) + Bu(t) + w(t), \qquad t\geq 0
   \label{eq: system}
\end{align}
with state $x(t) \in \mathbb{R}^n$, input $u(t) \in \mathbb{R}^m$, i.i.d. Gaussian noise $w(t) \sim \mathcal{N}(0, \Sigma)$, $\Sigma \succ 0$, system matrix $A \in \mathbb{R}^{n \times n}$, input matrix $B \in \mathbb{R}^{n \times m}$, and initial state $x_0 \sim \mathcal{N}(0, \Sigma_0)$, $\Sigma_0 \succ 0$. Further, define
\begin{equation}
\begin{aligned}
u_T &\triangleq \begin{bmatrix}
u(T-1)^\top & \cdots & u(0)^\top
\end{bmatrix}^\top && \in \mathbb{R}^{mT}, \\
w_T &\triangleq \begin{bmatrix}
w(T-1)^\top & \cdots & w(0)^\top
\end{bmatrix}^\top &&\in \mathbb{R}^{nT}, \\
C_T &\triangleq \begin{bmatrix}
B & AB & A^2 B & \cdots & A^{T-1} B
\end{bmatrix} && \in \mathbb{R}^{n\times mT}, \\
N_T &\triangleq \begin{bmatrix}
I & A & A^2 & \cdots & A^{T-1} 
\end{bmatrix} &&\in \mathbb{R}^{n \times nT}
\end{aligned}
\end{equation}
as the input sequence, noise sequence, $T$-step controllability matrix, and noise propagation matrix, respectively. Then, using model \eqref{eq: system}, the state at time $T$ is given as
\begin{equation}
    x(T) = A^Tx(0) + C_Tu_T + N_Tw_T.
\end{equation}
Let $x_f$ denote the final deterministic desired state at time $T$. The SMEC problem asks for an open-loop input sequence $u_T \in \mathbb{R}^{mT}$ that steers the system state \emph{close} to $x_f$ in time $T$ and minimizes the cost  
$\sum_{t=0}^{T-1} u(t)^\top Q_0 u(t) = u_T^\top Q_Tu_T$, where $Q_0 \succ 0$ and $Q_T = I_T \otimes Q_0$. To quantify the \emph{closeness} of $x_T$ to $x_f$, we define the error vector $e_T = x_T - x_f$. Since $w(t)$ is Gaussian, the error is also Gaussian with mean and covariance as
\begin{align*}
    \mu_{e_T} &= C_Tu_T-x_f, \text{ }\Sigma_{e_T} = A^T \Sigma_0 (A^T)^\top + N_T \Sigma_{w_T}N_T^\top,
\end{align*}
where $\Sigma_{w_T} \triangleq (I_T \otimes \Sigma)$. Since Rank$(N_T) = n$ and $\Sigma_{w_T} \succ 0$, we have $\Sigma_{e_T} \succ 0$. Hence, $\Sigma_{e_T}$ is invertible. We emphasize that the control input $u_T$ only affects the mean of the error and not the covariance. Next, we discuss various ways to mathematically constrain the \emph{size} of the random error $e_T$.
\subsection{Constraints on Normalized Error}
  We consider the normalized error $\Sigma_{e_T}^{-1/2}e_T$, and enforce it to be \emph{small}. Normalization accounts for the covariance of the error by assigning small weights to error components with large variances and vice versa. Next, we present three approaches to constrain the norm of the normalized error.\\
\emph{1. Chance Constraint:} This captures that the probability of the norm of the normalized error being small is large. Formally, this is specified as
\begin{equation}
\begin{aligned}
\mathbb{P}\{e_T^\top \Sigma_{e_T}^{-1} e_T \leq \epsilon_1^2\}\ge 1- \delta,
\label{eq: original}
\end{aligned}
\end{equation}
where $\epsilon_1>0$ and  $0\leq \delta\leq1$ are known constants. The choice of $(\epsilon_1,\delta)$ depends on the desired risk tolerance of the application.

The quantity $e_T^\top \Sigma_{e_T}^{-1}e_T$ has a non-central Chi-square distribution $(e_T^\top \Sigma_{e_T}^{-1}e_T \sim \chi_n^2(\lambda))$ with degrees of freedom $n$ and non-centrality parameter
$\lambda = \mu_{e_T}^\top \Sigma_{e_T}^{-1}\mu_{e_T}$. Thus, \eqref{eq: original} can be re-written in terms of the CDF of non-central Chi-square distribution $F_{\chi_n^2(\lambda)}$ as
\begin{equation}
    F_{\chi_n^2(\lambda)}(\epsilon_1^2) \geq 1- \delta.
    \label{eq: equi}
\end{equation}
For a given $\epsilon_1$ and $n$, $F_{\chi_n^2(\lambda)}(\epsilon_1^2)$ is a monotonically decreasing function of $\lambda$. Further, let $\underline{\gamma}_{T} \triangleq \underset{u_T}{\min} \:\: \mu_{e_T}^\top \Sigma_{e_T}^{-1}\mu_{e_T}.$ Then, constraint \eqref{eq: equi} is feasible iff $ F_{\chi_n^2(\underline{\gamma}_{T})}(\epsilon_1^2) \geq 1- \delta$. We provide the expression of $\underline{\gamma}_{T}$ later in Lemma \ref{lem: feasibility}. 
Finally, using the same monotonically decreasing property of $F_{\chi_n^2(\lambda)}(\epsilon_1^2)$, \eqref{eq: equi} is equivalent to
\begin{equation}
\begin{aligned}
     \lambda = \mu_{e_T}^\top \Sigma_{e_T}^{-1}\mu_{e_T} 
     \leq \lambda^*, 
     \label{eq: exact}
\end{aligned}     
\end{equation}
where $\lambda^*$ satisfies $F_{\chi_n^2(\lambda^*)}(\epsilon_1^2) = 1- \delta$.
This equation can be solved numerically to obtain the value of $\lambda^{*}$.  \\
From \eqref{eq: equi}, the minimum feasible value of $\delta$ is $1-F_{\chi_n^2(\lambda)}(\epsilon_1^2)$. This value increases with the system dimension $n$. Thus, satisfying the error constraint becomes difficult for larger systems.

\emph{2. Second-Order Moment Constraint:} 
We consider the expected value of the norm of the normalized error as
\begin{align*}
    \mathbb{E}[e_T^\top \Sigma_{e_T}^{-1} e_T] = \mu_{e_T}^\top \Sigma_{e_T}^{-1} \mu_{e_T} +n.
\end{align*}
The first component $\mu_{e_T}^\top \Sigma_{e_T}^{-1} \mu_{e_T}$ depends on $u_T$ and we constrain it as 
\begin{align}
    \mu_{e_T}^\top \Sigma_{e_T}^{-1} \mu_{e_T} & \leq \epsilon_2^2.
    \label{eq: second}
\end{align}
The above condition is feasible when  $\epsilon_2^2 \geq \underline{\gamma}_T$.\\
\emph{3. Mahalanobis Distance Constraint:} The Mahalanobis distance between a point $v$ and a distribution with mean $\mu_x$ and covariance matrix $\Sigma_x$ is given by $\sqrt{(v-\mu_x)^\top \Sigma_x^{-1}(v-\mu_x)}$. Thus, we bound the Mahalanobis distance between the normalized error and $0$ as
\begin{equation}\mu_{e_T}^\top \Sigma_{e_T}^{-1}\mu_{e_T} \leq \epsilon_3^2.
    \label{eq: maha}
    \end{equation}
The above condition is feasible when $\epsilon_3^2 \geq \underline{\gamma}_T$.

Observe that the constraints \eqref{eq: exact}, \eqref{eq: second}, and \eqref{eq: maha} are similar, and using $\mu_{e_T} = C_T u_T-x_f$, we write them in a general form as
\begin{equation}
    (C_Tu_T-x_f)^\top \Sigma_{e_T}^{-1}(C_Tu_T-x_f) \leq \gamma, 
    \label{eq: general}
\end{equation}
where $\gamma \geq\underline{\gamma}_T$. 
Using this, the SMEC problem is specified as the following Quadratically-Constrained Quadratic Program (QCQP)
\begin{equation}
\begin{aligned}
u_T^\star 
\;=\;
\underset{u_T \in \mathbb{R}^{mT}}{\arg \min} \quad
&  u_T^\top Q_T u_T \\
\text{s.t.} \quad & (C_Tu_T-x_f)^\top \Sigma_{e_T}^{-1}(C_Tu_T-x_f) \leq \gamma.
\label{eq: meci}
\end{aligned}
\end{equation}
Since $Q\succ 0$ and $\Sigma_{e_T}\succ0$, the above QCQP is convex.
\subsection{Constraints on Absolute Error}
In this subsection, we consider the absolute error $e_T$ and enforce
it to be \emph{small} in the following ways.\\
\emph{1. Chance Constraint:}
Similar to before, the chance constraint is specified as
\begin{equation}
    \mathbb{P}\{e_T^\top e_T \le \epsilon_4^2\} \ge 1-\delta_1,
    \label{eq: abs_original}
\end{equation}
with $\epsilon_4 > 0$ and $0 \le \delta_1 \le 1$. 
The quantity $e_T^\top e_T$ has a Generalized Chi-square distribution and the constraint \eqref{eq: abs_original} cannot be represented as a tractable constraint on $u_T$. Thus, we enforce a condition that guarantees \eqref{eq: abs_original} in the following two ways. First, using 
$\frac{1}{\lambda_{\text{max}}(\Sigma_{e_T})}e_T^\top e_T \le e_T^\top \Sigma_{e_T}^{-1}e_T \le \frac{1}{\lambda_{\text{min}}(\Sigma_{e_T})}e_T^\top e_T$, we have that the following constraint implies that \eqref{eq: abs_original} holds true
\begin{align*}
    \mathbb{P}\left\{e^\top \Sigma_{e_T}^{-1}e \le \frac{1}{\lambda_{\text{min}}(\Sigma_{e_T})} \epsilon_4^2\right\} \ge 1-\delta,
\end{align*}
which we solve as discussed in the previous section.
Other approaches include approximation using concentration inequalities. For example, the Chebyshev inequality results in the following conservative constraint
\begin{equation}
    \mu_{e_T}^\top \mu_{e_T} \le \left(\epsilon_4 - \sqrt{\frac{tr(\Sigma_{e_T})}{\delta}}\right)^2.
    \label{eq: chebyshev}
\end{equation}
\emph{2. Second-Order Moment Constraint:} We consider the expected value of the norm of the absolute error as
\begin{align*}
    \mathbb{E}[e_T^\top  e_T] = \mu_{e_T}^\top \mu_{e_T} + \text{trace}(\Sigma_{e_T}).
\end{align*}
Only the first component $\mu_{e_T}^\top \mu_{e_T}$ depends on $u_T$, and we constrain it as
\begin{align}
    \mu_{e_T}^\top\mu_{e_T} & \le \epsilon_5^2.
    \label{eq: abs_second}
\end{align}
Observe that absolute error constraints \eqref{eq: chebyshev} and \eqref{eq: abs_second} are similar. We use them to formulate the following QCQP:
\begin{equation}
\begin{aligned}
u^\star_T 
=
\underset{u_T \in \mathbb{R}^{mT}}{\arg \min} \quad
&  u_T^\top Q_T u_T \\
\text{s.t.} \quad & (C_Tu_T-x_f)^\top(C_Tu_T-x_f) \leq \gamma.
\label{eq: abs_meci}
\end{aligned}
\end{equation}
Next, we present the solution to Problem \eqref{eq: meci} and \eqref{eq: abs_meci}.
\section{Solution to Stochastic Minimum-Energy Control}
In this section, we address the optimization Problem \eqref{eq: meci} with the normalized error. Problem \eqref{eq: abs_meci} is similar and all its results can be obtained by replacing $\Sigma_{e,T}=I$. We begin by deriving the feasibility condition for Problem \eqref{eq: meci}.
\begin{lemma}[\emph{Feasibility}]\label{lem: feasibility}
There exists a feasible control sequence $u_T \in \mathbb{R}^{mT}$ for Problem \eqref{eq: meci} if and only if
\begin{equation}
    \gamma \geq \underline{\gamma}_{T} \triangleq 
    x_f^\top \left(\Sigma_{e_T}^{-1}
    - \Sigma_{e_T}^{-1}
C_T\left(C_T^\top\Sigma_{e_T}^{-1}C_T\right)^\dagger
    C_T^\top\Sigma_{e_T}^{-1}\right)x_f.
    \label{eq: min_fes}
\end{equation}
\end{lemma}
\begin{proof}
Let $g(u)=(C_Tu-x_f)^\top\Sigma_{e_T}^{-1}(C_Tu-x_f)$. The constraint in \eqref{eq: meci} is feasible if and only if
$\gamma\ge \underset{u}{\min}\: g(u)\triangleq \underline{\gamma}_{T}^{\text{(nor)}}$. Note that $\underline{\gamma}_{T}^{\text{(nor)}}$ is the projection error of $x_f$ on Range$(C_T)$ with respect to the $\Sigma_{e_T}^{-1}$ weighted norm. The corresponding projection matrix is given by $ P = C_T(C_T^\top \Sigma_{e_T}^{-1}C_T)^\dagger C_T^\top \Sigma_{e_T}^{-1}.$
Then, we have
\begin{align*}
\underline{\gamma}_{T}^{\text{(nor)}}
&=
(x_f-P x_f)^\top
\Sigma_{e_T}^{-1}
(x_f-Px_f) \\
&=
x_f^\top
(I-P)^\top
\Sigma_{e_T}^{-1}
(I-P)x_f 
=
x_f^\top
\Sigma_{e_T}^{-1}
(I-P)x_f, 
\end{align*}
where the last equality follows from  the symmetryidempotence ($P^2=P$) and self-adjointness ($P^\top\Sigma_{e_T}^{-1}= \Sigma_{e_T}^{-1}P$) of $P$.
Substituting $P$, we get the result. 
\qed
\end{proof}
Note that $\underline{\gamma}_{T}\geq 0$, and $\underline{\gamma}_{T} = 0 \Leftrightarrow x_f \in \text{Range}(C_T)$.\\
Next, we present the solution to the QCQP \eqref{eq: meci}.
\begin{theorem}[\emph{Optimal Solution}]
\label{thm: optimal_cost}
Assume $\gamma \geq \underline{\gamma}_{T}$. Define
\begin{align}
    S_T \triangleq \frac{\Sigma_{e_T}}{\nu_T} + C_TQ_T^{-1}C_T^\top,
\end{align}
and $\nu^*_T > 0$ be the unique solution to the following equation
\begin{align}
f(\nu_T) &\triangleq \frac{1}{\nu_T^{2}}x_f^\top S_T^{-1}\Sigma_{e_T} S_T^{-1} x_f = \gamma.
\label{eq: nu*}
\end{align}
Further, define
\begin{align}
C_{T,s} &\triangleq \big[ Q_T^{-1/2}C_T^\top S_T^{-1}\ \big]^\dagger Q_T^{1/2}, \quad 
W_{T,s} \triangleq 
C_{T,s}Q_T^{-1}C_{T,s}^\top.
\end{align}
Then, the solution to QCQP \eqref{eq: meci} is
\begin{equation}
u_T^* =
\left\{
\begin{aligned}
&0, \quad \text{if } \gamma \geq x_f^\top \Sigma_{e_T}^{-1} x_f,\\[3mm]
&Q_T^{-1}C_{T,s}^{\top} W_{T,s}^\dagger x_f,
\quad \text{if } \underline{\gamma}_{T} \le \gamma < x_f^\top \Sigma_{e_T}^{-1} x_f.\\[3mm]
\end{aligned}
\right.
\label{eq: solution}
\end{equation}
Further, the optimal cost is given by
\begin{align}
   J^*_T &= x_f^\top W_{T,s}^\dagger x_f
    = x_f^\top S_T^{-1}(C_TQ^{-1}_TC_T^\top)S_T^{-1}x_f. 
    \label{eq: optimal_cost}
\end{align}
\end{theorem}
\begin{proof}
If $x_f^\top \Sigma_{e_T}^{-1} x_f \geq \gamma$, then \(u_T=0\) is feasible for \eqref{eq: meci} and yields cost \(u_T^\top Q_T u_T = 0\). Since the cost is nonnegative, it follows that \(u_T^* = 0\).
Next, we consider the case $\underline{\gamma}_{T} \leq \gamma < x_f^\top \Sigma_{e_T}^{-1}x_f$ and use KKT conditions to solve \eqref{eq: meci}. Consider the Lagrangian
\begin{align}
  \label{eq:lagrangian}
    \mathcal{L}(u_T, \nu_T) &= u_T^\top Q_Tu_T \\ &+ \nu_T [(C_Tu_T-x_f)^\top \Sigma_{e_T}^{-1}(C_Tu_T-x_f)-\gamma], \nonumber
\end{align}
where $\nu_T \geq 0$ is the Lagrange multiplier. Further, let $u_T^*$ be the optimal solution for QCQP \eqref{eq: meci} and $\nu_T^*$ be the corresponding lagrange multiplier in \eqref{eq:lagrangian}.

Clearly, $u_T^* = 0$ is infeasible for $\underline{\gamma}_{T} \leq \gamma < x_f^\top \Sigma_{e_T}^{-1}x_f$ which implies that $\nu^* \neq 0$, and hence $\nu_T^* > 0$. Hence, from complementary slackness, the constraint must be active, i.e.
\begin{equation}
     (C_Tu_T^*-x_f)^\top\Sigma_{e_T}^{-1}(C_Tu_T^*-x_f) = \gamma.
     \label{eq: short}
\end{equation}
Also, using stationarity KKT condition i.e. $\nabla _{u_T} \mathcal{L}(u_T,\nu_T) = 0$, we get
\begin{align}
\label{eq: optimal_u}
u_T^* &= (Q_T + \nu_T^*C_T^\top \Sigma_{e_T}^{-1}C_T)^{-1} \nu_T^* C_T ^\top \Sigma_{e_T}^{-1}x_f. \\
& \overset{(a)}{=} Q_T^{-1}C_T^\top \big(\frac{\Sigma_{e_T}}{\nu_T^*} + C_T Q_T^{-1}C_T^\top \big)^{-1}x_f
\!\overset{(b)}{=} Q_T^{-1}C_{T,s}^\top W_{T,s}^\dagger x_f, \nonumber
\end{align}
where $(a)$ follows from Woodbury identity and $(b)$ follows from the identity i.e. $B^\top (BB^\top)^\dagger = B^\dagger$. This proves \eqref{eq: solution}. Note that $u_T^*$ in \eqref{eq: optimal_u} is zero in two cases: either $x_f = 0$ or $x_f \in ker(W_{T,s}^\dagger)= ker(C_{T,s}^\top)$, i.e. if $x_f$ is orthogonal to $Range(C_{T,s})$. 
Next, to solve for  $\nu_T^* > 0$,
we substitute the expression for $u^*$ in \eqref{eq: short} and obtain
\begin{align}
C_T u_T^* -x_f
&= (C_TQ_T^{-1}C_T^\top) \big(\frac{\Sigma_{e_T}}{\nu_T^*} + C_TQ_T^{-1}C_T^\top \big)^{-1}x_f - x_f \nonumber \\
&= -\frac{\Sigma_{e_T}}{\nu_T^*}\underbrace{\big(\frac{\Sigma_{e_T}}{\nu_T^*} + C_TQ_T^{-1}C_T^\top \big)^{-1}}_{S_T(\nu_T^*)^{-1}}x_f 
\label{eq: woodbury}
\end{align}
Substituting \eqref{eq: woodbury} in \eqref{eq: short}, we obtain \eqref{eq: nu*}.
Further, to show that $\nu_T^*$ is unique, note that since $S_T\succ0$ and $\Sigma_{e_T}\succ0$, we have $x_f^\top S_T^{-1}\Sigma_{e_T}S_T^{-1}x_f>0$
for $x_f\neq0$. Hence, $f(\nu_T)$ is well-defined for $\nu_T>0$, and since $\frac{d f}{d\nu_T}
=
-\frac{2}{\nu_T^3}
x_f^\top S_T^{-1}\Sigma_{e_T}S_T^{-1}x_f
<0,$
$f(\nu_T)$ is strictly decreasing for $\nu_T > 0$. Hence the equation $f(\nu_T) = \gamma$ admits a unique solution $\nu_T^*$.
Finally, using \eqref{eq: optimal_u}, the closed form expression for optimal cost is
\begin{align}
u_T^{*\top}Q_Tu_T^*
&= (Q_T^{-1}C_{T,s}^\top W_{T,s}^\dagger x_f)^\top
Q_T
(Q_T^{-1}C_{T,s}^\top W_{T,s}^\dagger x_f) \nonumber\\
&=x_f^\top W_{T,s}^\dagger
(C_{T,s}Q_T^{-1}C_{T,s}^\top)
W_{T,s}^\dagger x_f \nonumber\\
&=x_f^\top W_{T,s}^\dagger W_{T,s}W_{T,s}^\dagger x_f
=x_f^\top W_{T,s}^\dagger x_f \nonumber\\
&\overset{(a)}{=}x_f^\top S_T^{-1}C_TQ_T^{-1}C_T^\top S_T^{-1}x_f,
\label{eq: simple_cost}
\end{align}
where $(a)$ follows from $
\left[D^\dagger(D^\dagger)^\top\right]^\dagger
= D^\top D,$
with $D=Q_T^{-1/2}C_T^\top S_T^{-1}.$
This proves \eqref{eq: optimal_cost}. \qed
\end{proof}
The solution to the SMEC problem requires solving \eqref{eq: nu*} and understanding the behavior of the function $f(\nu_T)$. We have $f(0) = x_f^\top \Sigma_{e_T}^{-1}x_f$. Further, $f(\nu_T)$ is a monotonically decreasing function (see the proof of Theorem \ref{thm: optimal_cost}). Next, we derive the asymptotic value of $f(\nu_T)$ as $\nu_T \rightarrow \infty$. To this aim, define $M_T \triangleq \Sigma_{e_T}^{-1/2}C_TQ_T^{-1}C_T^\top\Sigma_{e_T}^{-1/2}$. Then, $f(\nu_T)$ simplifies to
\begin{align}
f(\nu_T)
&=
x_f^\top\Sigma_{e_T}^{-1/2}
\left(I+\nu_TM_T\right)^{-2}
\Sigma_{e_T}^{-1/2}x_f 
\label{eq: f_simplified}
\end{align}
Using the fact $\underset{\nu_T \to \infty} {\lim}(I+\nu_TM_T)^{-2} = P_{\text{Null}(M_T)}$ where $P_{\text{Null}(M_T)}$ is the orthogonal projector onto $\text{Null}(M_T)$, we get
\begin{align*}
\underset{\nu_T \to \infty} {\lim} f(\nu_T) = x_f^\top\Sigma_{e_T}^{-1/2}
P_{\text{Null}(M_T)}
\Sigma_{e_T}^{-1/2}x_f  
\overset{(a)}{=} \underline{\gamma}_{T}
\end{align*}
where, defining $Y=C_T^\top\Sigma_{e_T}^{-1/2}$, step $(a)$ follows from $\text{Null}(M_T) = \text{Null}(Y^\top Q_T^{-1}Y) = \text{Null}(Y),$
and thus, $P_{\text{Null}(M_T)} = I - Y^\dagger Y= I- Y^\top (Y Y^\top)^\dagger Y$.

Fig. \ref{fig: f} shows the plot of $f(\nu_T)$. It follows that for $\gamma = x_f^\top \Sigma_{e_T}^{-1} x_f$, we have $\nu_T^* = 0$, and as $\gamma\rightarrow \underline{\gamma}_{T}$, we get $\nu_T^* \rightarrow \infty$. We can solve \eqref{eq: nu*} using standard numerical algorithms, for instance, the bijection method.
\begin{figure}[H]
    \centering
\includegraphics[width=\linewidth]{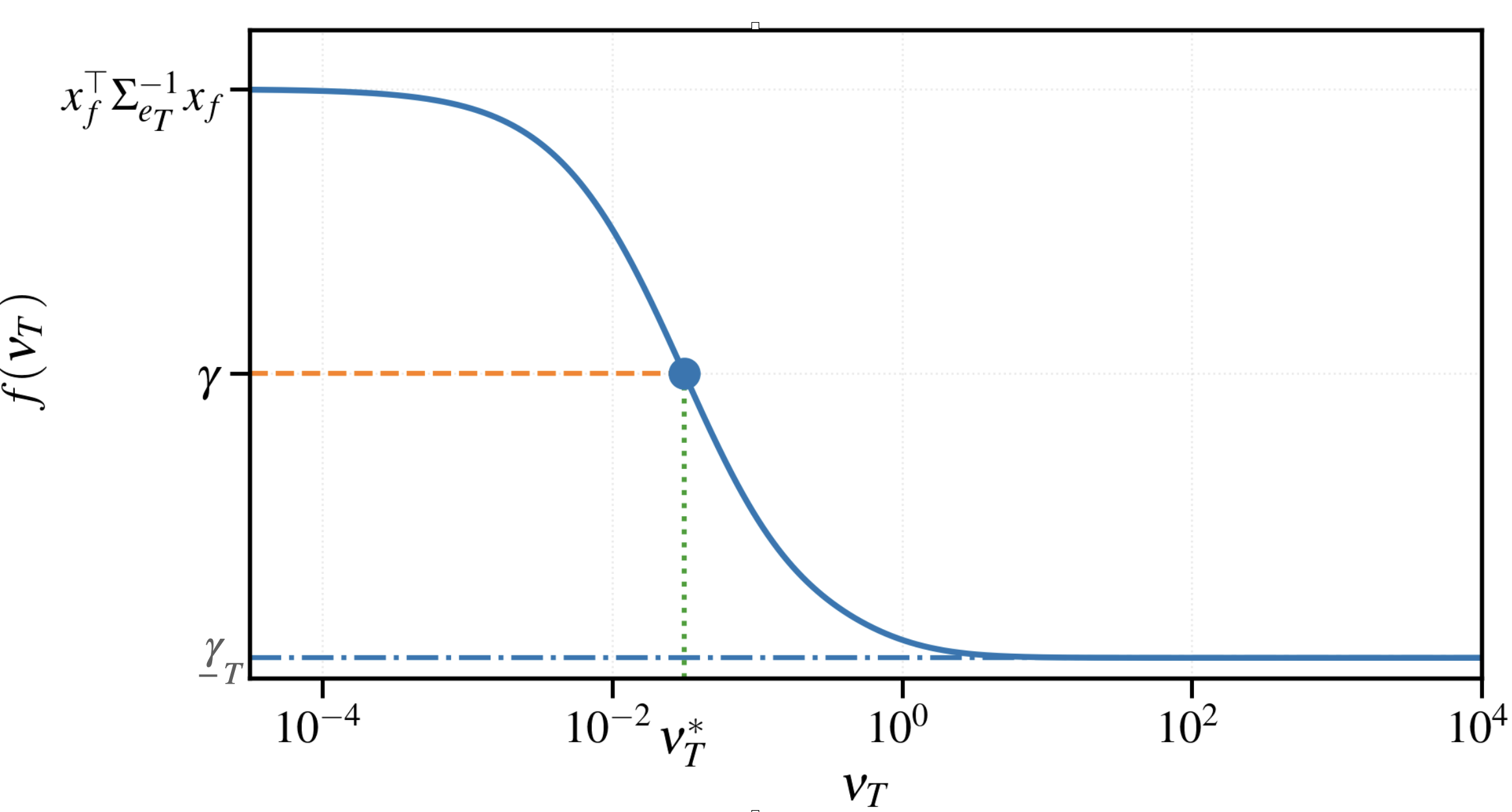}
    \caption{$f(\nu_T)$ vs $\nu_T$. Here $n= 5,$ $m=2,$ and $T=2$ with x-axis on log-scale.}
    \label{fig: f}
\end{figure}
 
Next, we simplify $f(\nu_T)$ further to explicitly reveal its dependence on the spectrum of $M_T$. To this aim, we define the eigen-decomposition of $M_T$ as $ M_T \!= V_T\Lambda_T V_T^\top$ with $\Lambda_T = \text{diag}(\lambda_1, \lambda_2, \cdots,\lambda_n).$
Substituting this in \eqref{eq: f_simplified}, we obtain $f(\nu_T) =\sum_{i=1}^{n}\frac{{\tilde{x}_{f_i}}^2}{(1+\nu_T \lambda_i)^2},$
where $\tilde x_f =  [{\tilde{x}_{f_1}}, \cdots, {\tilde{x}_{f_n}}]^\top \triangleq -V_T^\top \Sigma_{e_T}^{-1/2} x_f.$ This shows that constraint \eqref{eq: nu*} is equivalent to a polynomial equation of degree $2n$.
\subsection{Fundamental Connections between Stochastic and Deterministic MEC Problems}
Interestingly, the optimal control input and optimal cost for the SMEC problem have a similar structure to that of the Deterministic MEC (DMEC) problem, as explained next.\\
For $x_0=0$, the DMEC problem is formulated as
\begin{align*}
u_T^\star
\;=\;
&\underset{u_T\in\mathbb{R}^{mT}}{\arg\min}\quad
 u_T^\top Q_T u_T 
\\ & \qquad\text{s.t. } \:\:\:
C_Tu_T = P_{\operatorname{Range}(T)}x_f,
\label{eq: det_mec}
\end{align*}
where $P_{\operatorname{Range}(T)} = C_T C_T^{\dagger}$ is the orthogonal projector onto $\text{Range}(C_T)$. The above problem is always feasible. Its optimal control  input and optimal cost are given by
\begin{align}
u_T^\star
&=
Q_T^{-1}C_T^\top W_T^\dagger P_{\operatorname{Range}(T)} x_f  \overset{(a)}{=} Q_T^{-1}C_T^\top W_T^\dagger x_f, \nonumber \\
J_T^{\star} &= u_T^{\star\top}Q_Tu_T^\star
x_f^\top W_T^\dagger x_f,
\end{align}
where  $W_T = C_TQ_T^{-1}C_T^\top$ is the Gramian, and $(a)$ follows from $\text{Range}(C_T) = \text{Range}(W_T)$ which implies that $P_{\operatorname{Range}(T)} = W_T W_T^{\dagger}$. Crucially, the above solution is similar in structure to \eqref{eq: solution} and \eqref{eq: optimal_cost} with $W_{T,s}$ and $C_{T,s}$ as stochastic equivalents of $W_T$ and $C_T$, respectively.

Next, we show that SMEC problem reduces to DMEC problem under certain conditions.
\begin{enumerate}
    \item When $\gamma \to \underline{\gamma}_{T}$, we have $\nu_T^* \rightarrow \infty$ (c.f. discussion above Fig. \ref{fig: f}). If the system is controllable\footnote{However, if $\text{Rank}(C_T) < n$, then $\underset{\nu_T^* \rightarrow \infty}{\lim} C_{T,s} = C_T$, only if $U_r^\top \Sigma_{e_T}U_0 = 0$ where $U_r$ and $U_0$ are the orthonormal bases of $\text{Range}(C_T)$ and $\text{Null}(C_T^\top)$, respectively.} ($\text{Rank}(C_T) = n$), then $\underset{\nu_T^* \rightarrow \infty}{\lim} C_{T,s} = C_T$, and the solution to the SMEC problem coincides with the DMEC problem.
This also holds true for problem \eqref{eq: abs_meci} with the constraint on the absolute error.

\item When  $\Sigma_{e_T} \to \textbf{0}$ and $x_f\in \text{Range}(C_T)$, we get 
$\underset{\Sigma_{e_T} \to \textbf{0}}{\lim} C_{T,s} = C_T$, and the solutions coincide. Finally, if $x_f \notin \text{Range}(C_T)$ and $\Sigma_{e_T} \to \textbf{0}$, then the SMEC problem becomes infeasible. 
\end{enumerate}
This analysis shows that the SMEC problem is a generalization of the DMEC problem, and it reduces to DMEC problem under above conditions. 
\subsection{Fundamental Limit on Optimal Cost}
In this subsection, we analyze the optimal cost as a function of time horizon $T$ and show that it asymptotically converges. This provides a fundamental lower bound beyond which the optimal cost cannot be reduced by increasing the horizon $T$.
We begin with a result that provides a condition for the decrease of $J_T^*$.
\begin{proposition}[\emph{Decrease of $J_T^*$ with $T$}]\label{prop: monotonicity} Let $J_T^* = u_T^{*\top} Q_T u_T^*$ be the optimal cost of Problem \eqref{eq: meci}. Then, if $A\Sigma_{0}A^\top + \Sigma \succeq \Sigma_{0}$, we have $J_{T+1}^* \leq  J_{T}^*$ for all $T\geq 1$.
\end{proposition}
\begin{proof}
    Consider QCQP \eqref{eq: meci} at time $T+1$ with input $u_{T+1} = \begin{bmatrix}
    u_T^{*\top} & 0
    \end{bmatrix}^\top \in \mathbb{R}^{m(T+1)}$. Then, we have 
    \begin{align*}
         &(C_{T+1}u_{T+1}-x_f)^\top \Sigma_{e_{T+1}}^{-1}(C_{T+1}u_{T+1}-x_f) \\ 
         &= (C_Tu_T^* - x_f)^\top (A\Sigma_{e_T}A^\top + \Sigma)^{-1}(C_Tu_T^*-x_f)
    \end{align*}
Hence, if $A\Sigma_{e_T}A^\top + \Sigma \succeq \Sigma_{e_T}$ i.e. $A^T(A \Sigma_0 A^\top + \Sigma -\Sigma_0)(A^T)^\top \succeq 0 \Leftrightarrow A \Sigma_0 A^\top + \Sigma \succeq \Sigma_0$, $u_{T+1}$ is feasible at time $T+1$ which implies
$J(u_{T+1}^*) \leq J(u_{T+1}) = u_{T+1}^\top Q_{T+1} u_{T+1} = u_{T}^{*\top} Q_{T} u_{T}^* 
= J(u_T^*)$.
This completes the proof.
\qed
\end{proof}
Next, we establish the asymptotic convergence of the optimal cost.
\begin{theorem}[\emph{Convergence of Optimal Cost}]
\label{thm:horizon_limit}
Let $P$ be an invertible matrix such that
\[
P^{-1}AP=
\text{diag}(A_n, A_s), \:\:
|\lambda_i(A_n)|\geq1, \:\:
|\lambda_i(A_s)|<1.
\]
Define 
\[
P^{-1}B=
\begin{bmatrix}
B_n\\
B_s
\end{bmatrix},
P^{-1}x_f=
\begin{bmatrix}
x_f^n\\
x_f^s
\end{bmatrix},
P^{-1}\Sigma P^{-\top}
=
\begin{bmatrix}
\Sigma_{nn} & \Sigma_{ns}\\
\Sigma_{sn} & \Sigma_{ss}
\end{bmatrix},
\]

\begin{align*}
W_c^s
&\triangleq
\sum_{k=0}^{\infty}
A_s^kB_sQ_0^{-1}B_s^\top(A_s^k)^\top, \quad
W_w^s
\triangleq
\sum_{k=0}^{\infty}
A_s^k\Sigma_{ss}(A_s^k)^\top,
\end{align*}
and $\gamma_\infty
\triangleq
(x_f^s)^\top(W_w^s)^{-1}x_f^s.$ Further, for $0<\gamma<\gamma_\infty$, let $\nu_\infty>0$ be the unique
solution of
\begin{equation}
f_\infty(\nu) = (x_f^s)^\top
(W_w^s+\nu W_c^s)^{-1}
W_w^s
(W_w^s+\nu W_c^s)^{-1}
x_f^s
=
\gamma,
\label{eq:nu_limit}
\end{equation}
and let $S_\infty^s
\triangleq
W_c^s+\frac{W_w^s}{\nu_\infty}.$
Then, we have
\begin{equation}
\lim_{T\to\infty}J_T^*
=
\begin{cases}
(x_f^s)^\top
(S_\infty^s)^{-1}
W_c^s
(S_\infty^s)^{-1}
x_f^s, \:\text{ if }
0<\gamma<\gamma_\infty,\\[2mm]
0,
\quad \qquad \qquad \qquad \qquad \quad \:\: \quad \text{ if }\gamma>\gamma_\infty \nonumber.
\end{cases}
\label{eq:J_limit}
\end{equation}
\end{theorem}

\begin{proof}
Define $G_T
\triangleq
C_TQ_T^{-1}C_T^\top,$
$\widetilde\Sigma_0
\triangleq
P^{-1}\Sigma_0P^{-\top},$ $\widetilde G_T
\triangleq P^{-1}G_TP^{-\top},$ $\widetilde{M}_T
\triangleq
P^{-1}\Sigma_{e_T}P^{-\top},$
\begin{align*}
\widetilde S_T(\nu)
&\triangleq
P^{-1}S_T(\nu)P^{-\top}
\triangleq
\begin{bmatrix}
\mathcal A_T(\nu) & \mathcal B_T(\nu)\\
\mathcal B_T(\nu)^\top & \mathcal D_T(\nu)
\end{bmatrix}, \text{and} \nonumber\\
\widetilde z_T
&\triangleq
\begin{bmatrix}
z_{n,T}^\top & z_{s,T}^\top
\end{bmatrix}^{\top}
=
\widetilde S_T(\nu_T)^{-1}
\begin{bmatrix}
(x_f^n)^\top & (x_f^s)^\top
\end{bmatrix}^{\top}
\\
&=
\begin{bmatrix}
\mathcal A_T^{-1}x_f^n
-
\mathcal A_T^{-1}\mathcal B_T
\mathcal H_T^{-1}
\left(
x_f^s-\mathcal B_T^\top\mathcal A_T^{-1}x_f^n
\right)
\\[3mm]
\mathcal H_T^{-1}
\left(
x_f^s-\mathcal B_T^\top\mathcal A_T^{-1}x_f^n
\right)
\end{bmatrix},
\end{align*}
where $\mathcal H_T(\nu)
\triangleq
\mathcal D_T(\nu)
-
\mathcal B_T(\nu)^\top
\mathcal A_T(\nu)^{-1}
\mathcal B_T(\nu).$
We then re-write optimal cost as 
\begin{align}
J_T^*
&=
z_{n,T}^\top G_{nn,T}z_{n,T}
+
2z_{n,T}^\top G_{ns,T}z_{s,T}
+
z_{s,T}^\top G_{ss,T}z_{s,T}.
\label{eq: J_lim}
\end{align}
Using Lemma~\ref{lem:block_separation}, for every fixed $\nu>0$, $\mathcal H_T(\nu)
\to
W_c^s+\frac{1}{\nu}W_w^s,$ $\mathcal B_T(\nu)^\top
\mathcal A_T(\nu)^{-1}x_f^n
\to0.$
Therefore, $z_{s,T}(\nu)
\to
z_{s,\infty}(\nu)
\triangleq
\left(
W_c^s+\frac{1}{\nu}W_w^s
\right)^{-1}x_f^s.$

Since $M_{nn,T}\preceq \nu\mathcal A_T(\nu)$, with
$\mathcal A_T(\nu)^{-1}\to0$,
$\|\mathcal A_T(\nu)^{-1/2}\mathcal B_T(\nu)\|\to0$, and
$z_{s,T}$ is bounded, we obtain
\[
z_{n,T}^\top M_{nn,T}z_{n,T}
\le
\nu\left\|
\mathcal A_T(\nu)^{-1/2}
\bigl(x_f^n-\mathcal B_T(\nu)z_{s,T}\bigr)
\right\|^2
\to0.
\]
Further, since $\widetilde{M}_T\succeq0$, we have $|z_{n,T}^\top M_{ns,T}z_{s,T}|^2
\le
(z_{n,T}^\top M_{nn,T}z_{n,T})
(z_{s,T}^\top M_{ss,T}z_{s,T}),$ and hence
$z_{n,T}^\top M_{ns,T}z_{s,T}\to0.$
Therefore, for every fixed $\nu>0$,
\begin{align*}
f_T(\nu)
&=
\frac{1}{\nu^2}
z_T(\nu)^\top M_Tz_T(\nu)
\nonumber\\
&=
\frac{1}{\nu^2}
\Big(
z_{n,T}^\top M_{nn,T}z_{n,T}
+
2z_{n,T}^\top M_{ns,T}z_{s,T}
\\&+
z_{s,T}^\top M_{ss,T}z_{s,T}
\Big)
\nonumber\\
&\to
\frac{1}{\nu^2}
z_{s,\infty}(\nu)^\top
W_w^s
z_{s,\infty}(\nu)
\nonumber\\
&=
x_f^{s\top}
(W_w^s+\nu W_c^s)^{-1}
W_w^s
(W_w^s+\nu W_c^s)^{-1}
x_f^s 
= f_\infty(\nu),
\label{eq:f_block_expansion}
\end{align*}
where the last equality follows from $z_{s,\infty}(\nu)
=
\nu
(W_w^s+\nu W_c^s)^{-1}x_f^s.$
Hence, $f_\infty(\nu)$ is continuous and strictly decreasing, with $\lim_{\nu\to0}f_\infty(\nu)
=
(x_f^s)^\top(W_w^s)^{-1}x_f^s
=
\gamma_\infty,$ and $\lim_{\nu\to\infty}f_\infty(\nu)=0.$
Therefore, for $0<\gamma<\gamma_\infty$, using monotonicity of
$f_\infty$ and uniqueness of $\nu_T^*$, we obtain
$\nu_T\to\nu_\infty$ and $z_{s,T}\to
(S_\infty^s)^{-1}x_f^s.$

We next consider other terms in optimal cost. Since
$G_{nn,T}\preceq\mathcal A_T(\nu_T),$ we get $z_{n,T}^\top G_{nn,T}z_{n,T}\to0.$
Moreover, since $\widetilde G_T\succeq0$, we have $|z_{n,T}^\top G_{ns,T}z_{s,T}|^2
\le
(z_{n,T}^\top G_{nn,T}z_{n,T})
(z_{s,T}^\top G_{ss,T}z_{s,T}),$
and hence $z_{n,T}^\top G_{ns,T}z_{s,T}\to0.$
Substituting all these limits in \eqref{eq: J_lim}, we get the result.

Finally, let $\gamma>\gamma_\infty$. Then for every fixed $\nu>0$,
$f_\infty(\nu)<\gamma$, so for sufficiently large $T$,
$f_T(\nu)<f_T(\nu_T)=\gamma$. Since $f_T$ is decreasing,
$\nu_T<\nu$, and hence $\nu_T\to0$. Further, $\widetilde S_T(\nu_T)
\succeq\frac{1}{\nu_T}\widetilde M_T
\succeq\frac{1}{\nu_T}\widetilde\Sigma,$
where $\widetilde\Sigma=P^{-1}\Sigma P^{-\top}\succ0$. Thus, $0\le J_T^*
\le
\widetilde x_f^\top\widetilde S_T(\nu_T)^{-1}\widetilde x_f
\le
\nu_T\widetilde x_f^\top
\widetilde\Sigma^{-1}\widetilde x_f
\to0,$
which proves the second case of result.
\qed
\end{proof}
For the absolute-norm formulation, since \(\Sigma_{e_T}=I\), no asymptotic limit of \(\Sigma_{e_T}\) is required and the result follows directly by replacing \(W_w^s\) with \((P^{-1}P^{-\top})_{ss}\).
\section{Numerical Simulations}
In this section, we illustrate our results via numerical simulations.
We consider System \eqref{eq: system} with the following parameters
\begin{equation*}
\resizebox{\columnwidth}{!}{$
A=
\begin{bmatrix}
0.4 & 0.1 & 0.2 & 0.6\\
0.5 & -0.2 & 0.2 & -0.04\\
-0.02 & 0.1 & 0.04 & 0.4\\
0.2 & 0.03 & 0.1 & 0.09
\end{bmatrix},
\quad
B=
\begin{bmatrix}
1.4 & -0.2 & 0.3 & -0.8\\
-2.5 & 0.6 & 0.8 & -0.7\\
2.2 & -1.4 & 0.04 & -0.1\\
1.5 & 1.4 & 0.1 & 0.3
\end{bmatrix}
$},
\label{eq:system_matrices}
\end{equation*}
$\Sigma_0 = \sigma_0^2 I$, $\Sigma = \sigma^2 \hspace{1mm} I$, with $\sigma_0 = 0.2$ and $\sigma = 0.3$. The other parameters for the SMEC problem are $Q_0 = \operatorname{diag}(2,3,4,5)$, $x_f = [4, 5, 6, 7]^\top$, and
 $\epsilon_1 = \epsilon_2 = \epsilon_3 = 2.4$ in constraints \eqref{eq: original}, \eqref{eq: second}, \eqref{eq: maha}.
The minimum feasible value of $\delta$ for our example is $0.3$. For simulations, we consider $\delta = 0.4$.
To solve \eqref{eq: meci} under different constraints discussed in Section \ref{sec: problem}, we first obtain $F_{\chi_n^2(\lambda)}(\epsilon_1^2)$ using \emph{ncx2} function of scipy.stats library and then numerically obtain $\lambda^*=1.62$. Further, we use bisection method with $0$ and $100$ as lower and upper limits, respectively, to obtain optimal Lagrange multiplier $\nu_T^*$. 
Finally, we compute the optimal control sequence and optimal cost using the expressions in Theorem \ref{thm: optimal_cost}. For each optimal control input sequence, we perform $5000$ Monte-Carlo simulations and store the value of terminal states. 
\begin{figure}[!htbp]
\centering
\includegraphics[width=\linewidth]{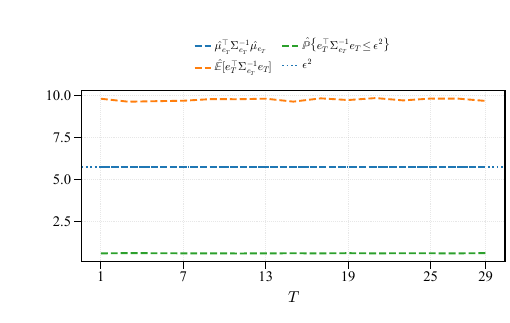}
  \caption{Empirical verification of constraints. Here, $\hat{\mathbb{P}}$ and $\hat{\mathbb{\mu}}_e$ denote the empirical estimates. $\hat{\mathbb{P}}\{e_T^\top \Sigma_{e_T}^{-1}e_T^\top \le \epsilon_1^2\}$, $\hat{\mathbb{E}}[e_T^\top\Sigma_{e_T}^{-1}e_T]$, and $\hat{\mu}_{e_T}^\top \Sigma_{e_T}^{-1}\hat{\mu}_{e_T}$ are close to $1-\delta, \epsilon_2^2 + n, \epsilon_3^2$, respectively.}
  \label{fig: rato}
\end{figure}

\begin{figure}[!t]
\includegraphics[width=\linewidth]{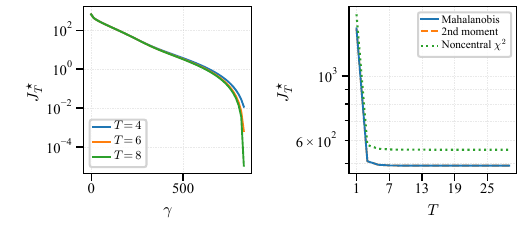}
    \caption{$J^*_T$ as a function of $\gamma$ and $T$.}
    \label{fig:x_gamma}
\end{figure}

\begin{figure}[!htbp]
\includegraphics[width=\linewidth]{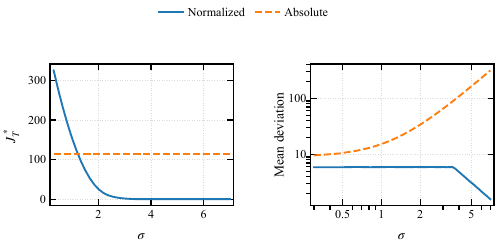}
  \caption{$J_T^*$ and Mean deviation as a function of $\sigma$. For absolute and normalized error norm formulations, we plot $\mu_{e_T}^\top \Sigma_{e_T}^{-1}\mu_{e_T}$ and $\mu_{e_T}^\top\mu_{e_T}+\text{trace}(\Sigma_{e_T}),$ respectively.}
  \label{fig: sigma}
\end{figure}
Fig. \ref{fig: rato} shows empirically that the constraints in \eqref{eq: original}, \eqref{eq: second} and \eqref{eq: maha} are satisfied with equality.

Fig. \ref{fig:x_gamma} shows the variation of the optimal cost with terminal deviation $\gamma$ and $T$. The optimal cost decreases with $\gamma$ since a larger value of $\gamma$ implies a larger feasible set. This indicates a trade-off between the cost and terminal deviation. Further, the optimal cost decreases with time since the condition in Proposition \ref{prop: monotonicity} is satisfied, and it asymptotically converges as per Theorem \ref{thm:horizon_limit}.

Finally, Fig. \ref{fig: sigma} shows the variation of optimal cost and mean deviation with process noise level for both normalized error formulation \eqref{eq: meci} and absolute error formulation \eqref{eq: abs_meci}. For the normalized error formulation, $J_T^*$ decreases with $\sigma$ due to expansion of the feasible set of $u_T$. The mean deviation $\mu_{e_T}^\top \Sigma_{e_T}^{-1}\mu_{e_T}$ remains equal to $\gamma$ until $u_T^*=0$. Subsequently, it decreases since the constraint is no longer active. For the absolute error norm formulation, $J_T^*$ remains constant since it does not depend on $\sigma$. The mean deviation $\mu_{e_T}^\top \mu_{e_T} + \text{trace}(\Sigma_{e_T})$ increases due to the increase of uncontrolled component $\text{trace}(\Sigma_{e_T})$.

\section{Conclusion}
We studied the open-loop MEC problem for stochastic linear systems. The study of this problem led us to propose the novel stochastic controllability matrix and stochastic Gramian. We conjecture that these quantities should play a role in other types of stochastic control problems as well. Future works include investigating this point, generalizing the problem to include state cost as well, and studying the problem in a data-driven setting where the system model is unknown.

\appendix
\begin{lemma}[\emph{Asymptotic Block Separation}]
\label{lem:block_separation}
Under the assumptions and notations of Theorem~\ref{thm:horizon_limit}, 
for every fixed $\nu>0$, as $T\rightarrow \infty$, we have 
\begin{align}
\mathcal A_T(\nu)^{-1} &\to 0, \text{ } \mathcal B_T(\nu)^\top\mathcal A_T(\nu)^{-1}
\mathcal B_T(\nu) \to0, \nonumber \\
\mathcal B_T(\nu)^\top\mathcal A_T(\nu)^{-1}x_f^n
&\to0, \text{ and }
\mathcal D_T(\nu)
\to
W_c^s+\frac{1}{\nu}W_w^s. 
\label{eq: limits}
\end{align}
\end{lemma}

\begin{proof}
Since $\widetilde\Sigma\succ0$ and
$|\lambda_i(A_n)|\ge1$, we have $\lambda_{\min}(M_{nn,T})\to\infty.$

Hence for $\mathcal A_T(\nu)
=
G_{nn,T}+\frac{1}{\nu}M_{nn,T}
\succeq
\frac{1}{\nu}M_{nn,T},$ we obtain $\mathcal A_T(\nu)^{-1}\to0$.
For $\rho(A_s)<1$, all stable-mode
Gramian series converge, and therefore $G_{ss,T}\to W_c^s,$ $M_{ss,T}\to W_w^s,$
Next, for fixed $K<T$,  we write $\mathcal B_T=\mathcal B_T^{(K)}+\mathcal R_T^{(K)}.$
Therefore $\|\mathcal A_T^{-1/2}\mathcal B_T\|
\le
\|\mathcal A_T^{-1/2}\mathcal B_T^{(K)}\|
+
\|\mathcal A_T^{-1/2}\mathcal R_T^{(K)}\|.$
$\mathcal B_T^{(K)}$ is bounded and
$\mathcal A_T^{-1}\to0$, hence the first term vanishes. Moreover, using 
positive semidefiniteness of the tail block, we get $\|\mathcal A_T^{-1/2}\mathcal R_T^{(K)}\|^2
\le
\|\mathcal D_{T,\mathrm{tail}}^{(K)}\|.$

Since $\rho(A_s)<1$, $\sup_{T>K}\|\mathcal D_{T,\mathrm{tail}}^{(K)}\|\to0
,(K\to\infty).$
Therefore, $\|\mathcal A_T^{-1/2}\mathcal B_T\|\to0.$ Finally, $\|\mathcal B_T^\top\mathcal A_T^{-1}x_f^n\|
\le
\|\mathcal A_T^{-1/2}\mathcal B_T\|
\|\mathcal A_T^{-1/2}x_f^n\|
\to0.$ This completes the proof.
\qed
\end{proof}




\bibliographystyle{IEEEtran}
\bibliography{references}
\end{document}